\documentclass[11pt,a4paper]{article}
\usepackage{a4wide,xspace}
\usepackage{fullpage,xspace,times}
\usepackage{graphicx,epsfig}
\usepackage{amsmath,upgreek}
\usepackage{amsfonts}
\usepackage{amssymb}
\usepackage{epstopdf}

\newtheorem{theorem}{Theorem}

\newtheorem{lemma}[theorem]{Lemma}

\newtheorem{proposition}[theorem]{Proposition}

\newcommand{\Xomit}[1]{ }
\newenvironment{proof}[1][Proof]{\textbf{#1.} }{\ \rule{0.5em}{0.5em}}

\mathchardef\mhyphen="2D

\newcommand{\eps}{\upvarepsilon}

\begin{document}

\title{Parallel solutions for ordinal scheduling  \\ with a small number of machines}
\date{}
\author{Leah Epstein\thanks{
Department of Mathematics, University of Haifa, Haifa, Israel.
\texttt{lea@math.haifa.ac.il}. }}
\maketitle



\begin{abstract}
We study ordinal makespan scheduling on small numbers of identical machines, with respect to two parallel solutions.
In ordinal scheduling, it is known that jobs are sorted by non-increasing sizes, but the specific sizes are not known in advance. For problems with two parallel solutions, it is required to design two solutions, and the performance of algorithm is tested for each input using the best solution of the two.
We find tight results for makespan minimization on two and three machines, and algorithms that have strictly better competitive ratios than the best possible algorithm with a single solution also for four and five machines. To prove upper bounds, we use a new approach of considering pairs of machines from the two solutions.
\end{abstract}

\section{Introduction}
Ordinal scheduling on $m\geq 2$ identical machines is defined as follows. The algorithm is given a certain number of jobs, $n\geq 1$, where $n$ is not known in advance. Job $j$ (for a positive integer $j$) has a size $p_j\geq 0$, which is unknown to the algorithm at the time of assignment. It is known in advance that the input is sorted by non-increasing size, in the sense that $p_j \leq p_{j+1}$ for any $j\geq 1$, where jobs that do not exist in the input are seen as jobs of size zero. The algorithm is required to have a good performance for any input satisfying the sorted order, where every such input is also called an input realization or realization. Thus, a solution is a partition of all positive integers (which serve as indexes of jobs) into $m$ sets, where the $i$th set is the subset of jobs assigned to machine $i$. We denote the $i$th set by $J_i$. Those sets are infinite but the length of the prefix of jobs of positive sizes is finite for every realization.
For a given realization, and a given solution, the load of machine $i$, $L_i$, is defined as $\sum_{j\in J_i} p_j$. The completion time for this solution is $\max_{1\leq i \leq m} L_i$. The goal is to minimize this cost or objective, that is, the goal is to minimize the completion time of the schedule. Ordinal algorithms for various problems (mainly scheduling problems) were studied in a number of articles
\cite{LSV96,liu96bp,liu1996ordinal,mahadev1998meaningfulness,tan01,yong2001semi,he02,tan2002ordinal,tan2004ordinal,tan05,ji2017ordinal,wang2019better,ELLMR22}.

This problem can be seen as a combinatorial problem of designing universal schedules. However, due to the relation to semi-online scheduling of sorted inputs (see below), and due to its algorithmic nature, we view the problem as an algorithmic one. We study ordinal algorithm as online (or semi-online) algorithms, and analyze them via the competitive ratio measure. The competitive ratio is the worst-case ratio between the cost of an ordinal solution and the cost an optimal offline solution, for the same realization. The cost for the standard variant of ordinal scheduling is the completion time of the solution.
We consider a fixed optimal offline solution, and assume that it always knows the realization, while an ordinal algorithm creates a solution without having this knowledge, and it has to create a schedule that is relatively good for all realizations. The work that is most relevant to our model is that of Liu et al. \cite{LSV96}, who studied the standard ordinal model defined above. In that study, upper and lower bounds were provided for the competitive ratio. For an arbitrary number of machines, the lower bound that was shown is $1.5$ (holding for any $m\geq 18$), and the upper bound is $\frac 53 \approx 1.66667$. More precisely, the upper bound is slightly smaller for every fixed value of $m$, and its supremum value is $\frac 53$. For $m=2,3,4$ machines, smaller bounds were shown. The upper bounds are $\frac 43 \approx 1.33333$, $1.4$, and $\frac{101}{70}\approx 1.4428571$, for $m=2$, $3$, and $4$, respectively. Lower bounds for $m=2,3$ were shown to be tight with the upper bounds, and a close lower bound of $\frac {23}{16}=1.4375$ on the competitive ratio was proved for $m=4$.

One idea of the algorithms used here, which was also used for example in \cite{LSV96,tan01,yong2001semi,he02,tan05,ji2017ordinal,wang2019better} (sometimes for other ordinal scheduling problems) is to assign different numbers of jobs to the different machines using a fixed pattern (but the numbers cannot be too different). The assignment of jobs of small indexes is more crucial than the assignment of jobs with large indexes, since the jobs presented later may be smaller. Moreover, if there is a large number of identical jobs,  the meaning is that none of the jobs is very large compared to the optimal offline completion time. As an example (which is based on the results of \cite{LSV96}), consider the case $m=2$, and the assignment of the first four jobs. If all those jobs have identical sizes, and the list of first four jobs is $<1,1,1,1>$, the assignment of three jobs to one machine leads to a competitive ratio of $1.5$. However, assigning two jobs to each machine leads to a competitive ratio of $\frac 43$ for the following list of sizes of the first four jobs $<3,1,1,1>$. In any algorithm with competitive ratio not exceeding $\frac 43$, the job that should be in the same set as job $1$ is job $4$, due to the following list of sizes for the first three jobs $<2,1,1>$, as otherwise, the competitive ratio is at least $1.5$. Thus, in order to obtain a competitive ratio of $\frac 43$, the subsets for the first four jobs are $\{1,4\}$ and $\{2,3\}$. The fifth job has to join the second set, due to the following list of sizes for the first five jobs: $<4,1,1,1,1>$. However, in our model which we explain below, in which the algorithm can construct two solutions, these arguments do not hold, and even an algorithm for two machines has to be designed very carefully, by finding the right balance between the two solutions.

We saw that even very simple realizations do not allow an algorithm to perform well. Here, we consider an approach that allows one to obtain a better performance for small numbers of machines, at the expense of using two solutions rather than one. A one-solution algorithm is an algorithm that maintains a single solution. A two-solution algorithm is an algorithm that maintains two solutions. One is allowed to use the better solution of the two for every input or realization. Thus, we define the makespan as the minimum completion time out of the two solutions. If the machines loads for one solution are denoted by $L^1_i$, and the machine loads for the other solution are denoted by $L^2_i$ (where the index $i$ is an index of a machine, and the loads are computed for a certain realization), the makespan is defined as $\min_{k=1,2}\{ \max_{1\leq i \leq m}L^k_i \}$.

Online problems with parallel solutions can be seen as an online problem with advice \cite{BKKKM17,BFKM18}, where the case of two solutions is equivalent to a single advice bit provided for the entire input.  There is work on problems with advice and parallel solutions for scheduling and related problems such as bin packing \cite{KKST97,Doh15,RRVS,BKLLO,Mik16,AH17_1,BKKKM17,ADKRR,BFKM18,FGK019}. It is typically hard to design good algorithms for just two solutions, in the sense that improving over the one-solution algorithm in a reasonably simple way often requires more than two solutions \cite{FGK019}. On the other hand, from a practical point of view, even if it is possible to maintain several solutions with the goal of using one of them (the best one) when the input is ended, this fixed number of solutions cannot be very large, and we believe that the case of two solutions is the most interesting one among studies of multiple solutions.


The difficulty in designing algorithms with multiple parallel solutions lies in the analysis. In the example above, it is obvious that by using two solutions, one can avoid the difficulty for very small inputs. Still, as we can see from our lower bound proofs, increasing input sizes slightly still leads to situations that none of the two solutions is perfect. For the upper bound proofs, we use the following method. For some machines we show that their loads are sufficiently small, no matter what the realization is. This analysis resembles the known results, and it relies on the sorted order. However, since a single solution cannot achieve the competitive ratios achieved here by two solutions, there are realizations for which one solution is not very good and other realizations for which the other solution is bad. We would like to show that at least one solution is good, and for that we analyze sums of loads (of the two solutions), or weighted sums. For small numbers of machines, the number of pairs to be considered will not be large.
The lower bounds constructions are obviously harder, since our algorithms perform well on very small inputs. In such proofs, it is also tricky to verify that indeed two solutions fail, but this needs to be done in order to obtain an understanding of the required action for the two solutions, in the algorithmic design.

Semi-online scheduling of jobs arriving over a list is a different though related online scheduling problem \cite{Gr69,SSW00,CKK}. In this problem, jobs arrive sorted by non-increasing sizes, but the algorithm knows the exact size for every arriving job. The objective is similar, but the difficulty of the online algorithm is in the fact that there is no information of future jobs (except for the sorted order). Every ordinal algorithm can obviously be used as a semi-online algorithm, but one can design better algorithms that use the actual sizes for their decisions. Graham \cite{Gr69} analyzed an algorithm called LPT (Longest Processing Time) which assigns each job in turn to the least loaded machine (breaking ties arbitrarily). He showed that the competitive ratio of LPT is $\frac {4m-1}{3m}=\frac 43-\frac{1}{3m}$. The problem was studied further by Seiden et al. \cite{SSW00}, who showed that LPT is the best possible semi-online algorithm for $m=2$, and proved a lower bound of $\frac{1 + \sqrt{37}}6 \approx 1.18046$ on the competitive ratio for three machines. Cheng et al. \cite{CKK} designed an algorithm of this last competitive ratio for three identical machines (which is a tight result), and an algorithm of competitive ratio at most $1.25$ for $m\geq 4$ machines. Thus, we can see that this variant is easier than the ordinal one. Here, we will show that even allowing the algorithm to use two solutions rather than one keeps the best competitive ratio higher for any $m$, when job sizes are not known in advance.

In this work, we study two-solution algorithms for small numbers of machines. We find tight bounds for $m=2,3$ and improved upper bounds on the competitive ratio for $m=4,5$, in the sense that they are better than the best possible bounds of one-solution ordinal algorithms, where the competitive ratios are at most $1.375$ and $1.4$, respectively. The tight bounds on the competitive ratio for $m=2,3$ are $1.25$ and $\frac 43$, respectively (this can be compared to the results $\frac 43$ and $1.4$, respectively, of a one-solution algorithm \cite{LSV96}). We show a lower bound of $\frac 43$ on the competitive ratio not only for $m=3$, but also for any $m\geq 4$. The one-solution known upper bounds for $m=4,5$ \cite{LSV96} are larger than our bounds (the bound is $1.5$ for $m=5$). A lower bound of $1.4375$ was known for a one-solution algorithm for $m=4$ (and this is higher than our upper bound of $1.375$ of a two-solution algorithm), but the known lower bound on the competitive ratio for $m=5$ was approximately $1.3704$ \cite{LSV96}. Here, we improve the lower bounds on the competitive ratio of a one-solution algorithm, and show that it is strictly above $1.4$ (the lower bound value that we show is $1.448598$). In fact, we find improved lower bounds on the competitive ratio of a one-solution algorithm for all values $5\leq m \leq 17$.

As another motivation for our work, we show in Section \ref{premi} that increasing the number of solutions will not allow us to design an algorithm that outputs an optimal solution for every realization. That is, we show that for every finite set of solutions, there exists an input realization such that none of the solutions creates an optimal schedule. This supports our thesis that the number of used solutions should be small both from a theoretical point of view and from a real-life perspective.

\section{Preliminaries}\label{premi}
In this section we show that increasing the number of solutions will not allow us to reach an optimal solution for $m=2$. We also define the notation, and discuss the approach of our proofs, and constraints on optimal solutions.

We start with the lower bound for a fixed number of solutions, $M$. For $M=2$, which is the topic of this work, we show an improved (and tight) lower bound later.

\begin{theorem}
For any constant number of parallel solutions, $M$, and $m=2$, the competitive ratio for the minimization problem  is strictly above $1$, and it is $1+\Omega(\frac{1}{M^2})$.
\end{theorem}
\begin{proof}
We say that a solution has type $k$ if $k\geq 2$ is the minimum integer such that job $k$ is not in the same set as job $1$. For example, if for the first ten jobs the subsets for the two machines are $\{1,2,3,6,7,10\}$ and $\{4,5,8,9\}$, we have $k=4$. Given $M$ solutions, by the pigeonhole principle, there is an integer $i$ in $\{3,\ldots,M+3\}$ such that there is no solution of type $i$ among the $M$ solutions (and we select such a value $i$).
Consider the following input that is based on the value of $i$.

The first $i-1$ jobs have sizes of $i$, and the next $i$ jobs have sizes of $i-1 \geq 2$. The remaining jobs have sizes of $0$ (so only $2i-1$ jobs have positive sizes). For this input, an optimal offline solution assigns the jobs of size $i$ to one machine and the jobs of size $i-1$ to the other machine, for a completion time of $i(i-1)$.

Consider all $M$ solutions, where the solutions are partitioned into solutions whose types are strictly larger than $i$ and types that are strictly smaller than $i$. For a solution of type $a\geq i+1$, the first $i$ jobs are in one set, and the maximum completion time is at least $(i-1)\cdot i + (i-1) = i^2 -1$.

For a solution of type $a \leq i-1$, the $a$th job is not in the same subset as the first job, so every machine has at least one job of size $i$ for the realization defined here. Let $\alpha$ and $\beta$ be the numbers of jobs of the two sizes $i$ and $i-1$ assigned to the first machine in this realization, that is, $\alpha$ is the number of jobs among the first $i-1$ jobs assigned by this solution to the first machine, and $\beta$ is the number of jobs assigned to the first machine among the next $i$ jobs. We know that $1\leq \alpha \leq i-2$, and $0\leq \beta \leq i$.

The only non-negative integer solutions $\alpha$ and $0\leq \beta \leq i$  of $\alpha\cdot i + \beta\cdot (i-1)=i(i-1)$ are $\alpha=i-1$, $\beta=0$ and $\alpha=0$, $\beta=i$. However, we assume that $1 \leq \alpha \leq i-2$, so are no such integers $\alpha$, $\beta$, and thus the completion time of the first machine is not $i(i-1)$. Therefore, for every solution among the $M$ solutions, the two machines have different completion times, and the maximum completion time is at least $i(i-1)+1$. Thus, by $i^2-1 \geq i^2-i+1$ for $i\geq 2$, we find that the competitive ratio is at least $\frac{i(i-1)+1}{i(i-1)}=1+\frac{1}{i(i-1)} \geq 1+\frac1{(M+2)(M+3)}$ (since $i\leq M+3$).
\end{proof}

We will always use $W$ to denote the total size of all jobs in the realization of the input. Thus, $W=\sum_{j=1}^{\infty} p_j$.
We use $\lambda$ to denote the optimal offline completion time for a specific realization. By averaging, we have $\lambda \geq \frac{W}m$, or alternatively, $W \leq m \cdot \lambda$.
We have $\lambda \geq p_1$, and since $p_1 \geq p_j$ holds for any $j\geq 2$, we also have $\lambda \geq p_j $ for any $j\geq 1$ (for example, we may use the property $\lambda \geq p_2$). Additionally, every solution assigns at least two jobs out of the first $m+1$ jobs to the same machine. Thus, $\lambda \geq p_m+p_{m+1}$, which implies $p_j \leq \frac{\lambda}2$ for $j\geq m+1$. We use $n$ to denote the number of jobs in a realization, in the sense that there are at most $n$ jobs whose sizes are positive.  For any $j\geq 0$, let $P_j=\sum_{\ell=1}^j p_j$ (so $P_0=0$), and for $j\geq 1$, $Q_j=\sum_{\ell=j}^n p_j$ (so $Q_n=p_n$ and $W=P_n=Q_1=P_j+Q_j-p_j$ for any $1\leq j \leq n$). Obviously, $P_j \geq j \cdot p_j$ since $p_1 \geq p_2 \ldots \geq p_j$.

For our two solutions, we will denote the loads of machine $i$ (for $1\leq i \leq m$) by $C_i$ and $L_i$.
In the analysis, we sometimes use expressions of the form $\ell\cdot C_i + \sum_{j\in J} t_j\cdot p_j$ or $\ell\cdot L_i + \sum_{j\in J} t_j\cdot p_j$, where the values $t_j$ are integers (which may be positive or negative or equal to zero), and $\ell$ is a positive integer.


Our methodology for proving a upper bound of $R$ for the competitive ratio of a two-solution algorithm is as follows. Given two solutions, the first step is to find an upper bound of $R \cdot \lambda$ for $C_{i_1}$ for a large set of integers $i_1$, and to find an upper bound of $R \cdot \lambda$ for $L_{i_2}$ for a large set of integers $i_2$. Denoting those sets by $I_1,I_2 \subseteq \{1,2,\ldots,m\}$, respectively, we are left with the complement sets $U_1, U_2$ (where $U_1 \cup I_1 = U_2 \cup I_2 = \{1,2,\ldots,m\}$, and $U_q \cap I_q =\emptyset$ for $q=1,2$. Then, for every pair $i'_1 \in U_1$, $i'_2\in U_2$, we show that there exists a linear combination with two integers $k_1,k_2 \geq 1$, such that $k_1\cdot C_{i'_1}+ k_2\cdot C_{i'_2} \leq (k_1+k_2)\cdot R \cdot \lambda$. The pair of values $k_1,k_2$ may be different for different pairs of machines. In this way, using the pigeonhole principle, we will find that it either holds $\max_{i'_1 \in U_1} C_{i'_1} \leq R \cdot \lambda$ or $\max_{i'_2 \in U_2} C_{i'_2} \leq R \cdot \lambda$ (or both). Since the loads of machines in $I_1$ do not exceed $R \cdot \lambda$ for the first solution, and the loads of machines in $I_2$ do not exceed $R \cdot \lambda$ for the second solution, we will find that for any input, at least one of the two solutions satisfies the competitive ratio. Obviously, the identity of the better solution is based on the realization of the input.

Consider a set of jobs whose indexes are of the form $\alpha \cdot k + \beta$, that consists exactly of all such jobs, for all $k\geq \gamma$, where $\alpha,\beta,\gamma$ are fixed integers, $0 \leq \beta < \alpha$, and $\alpha , \gamma \geq 1$. For example, if $\alpha=5$, $\beta=2$, and $\gamma=1$, the job set is $\{7,12,17,22,\ldots\}$. It is required that $\alpha\cdot \gamma +\beta - \alpha \geq 0$ will be valid, which holds by $\alpha\cdot \gamma +\beta - \alpha \geq \alpha+\beta-\alpha \geq 0$.

\begin{lemma}\label{theW}
For the set of jobs $\{\alpha \cdot \gamma + \beta, \alpha \cdot (\gamma+1) + \beta, \ldots\}$, the total size is at most $\frac{W-P_{\alpha \cdot (\gamma-1) + \beta}}{\alpha}$.
\end{lemma}
\begin{proof}
By $\alpha \cdot (\gamma-1)+\beta \geq 0$, the expression in the statement of the lemma is well-defined.
Consider the suffix of jobs, starting with job $\alpha \cdot (\gamma-1)+\beta+1$. The total size of these jobs is $Q_{\alpha \cdot (\gamma-1)+\beta+1}=W-P_{\alpha \cdot (\gamma-1)+\beta}$. We have $Q_{\alpha \cdot (\gamma-1)+\beta+1}=\sum_{j=\alpha \cdot (\gamma-1)+\beta+1}^n p_j$. Let $N\leq n$ be the largest value of the form $\alpha \cdot k + \beta$ for any integer $k$, where
$N=\alpha \cdot k' + \beta$. It holds that $Q_{\alpha \cdot (\gamma-1)+\beta+1} \geq \sum_{j=\alpha \cdot (\gamma-1)+\beta+1}^N p_j=\sum_{\kappa=\gamma}^{k'} \sum_{a=1}^{\alpha} p_{\alpha\cdot (\kappa-1)+a+\beta} \geq \sum_{\kappa=\gamma}^{k'} \sum_{a=1}^{\alpha} p_{\alpha\cdot (\kappa-1)+\alpha+\beta}=\alpha \cdot \sum_{\kappa=\gamma}^{k'} p_{\alpha\cdot \kappa+\beta}$. The last expression is $\alpha$ times the total size we are interested in.
\end{proof}

\section{Algorithms with two solutions}
In this section we design two-solution algorithms for $m=2,3,4,5$. For each case, we define the two solutions using disjoint sets whose union is the set of positive intgers.
Below, we define four algorithms with two solutions each. The two solutions for each one of the cases are also shown in Figures \ref{ff1} and \ref{ff2}.

{\begin{figure}[h!]
\centerline{{\epsfig{file=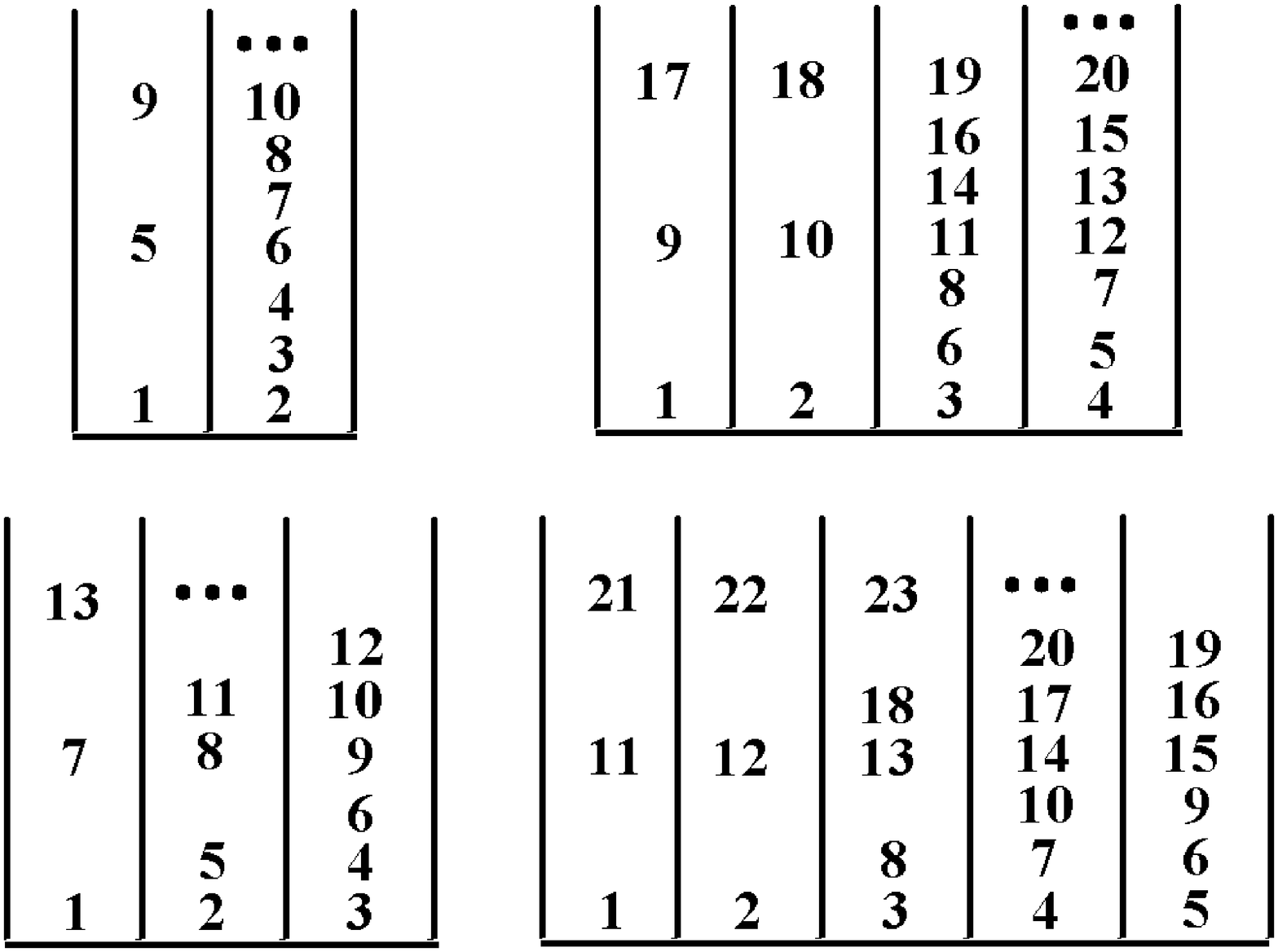, scale=0.33,angle=90}}}
\caption{The assignment of the first few jobs for the first solution for each number of machines. The top left solution is for two machines, the bottom left solution is for three machines, the top right solution is for four machines, and the bottom right solution is for five machines.}
\label{ff1}
\end{figure}}

{\begin{figure}[h!]
\centerline{{\epsfig{file=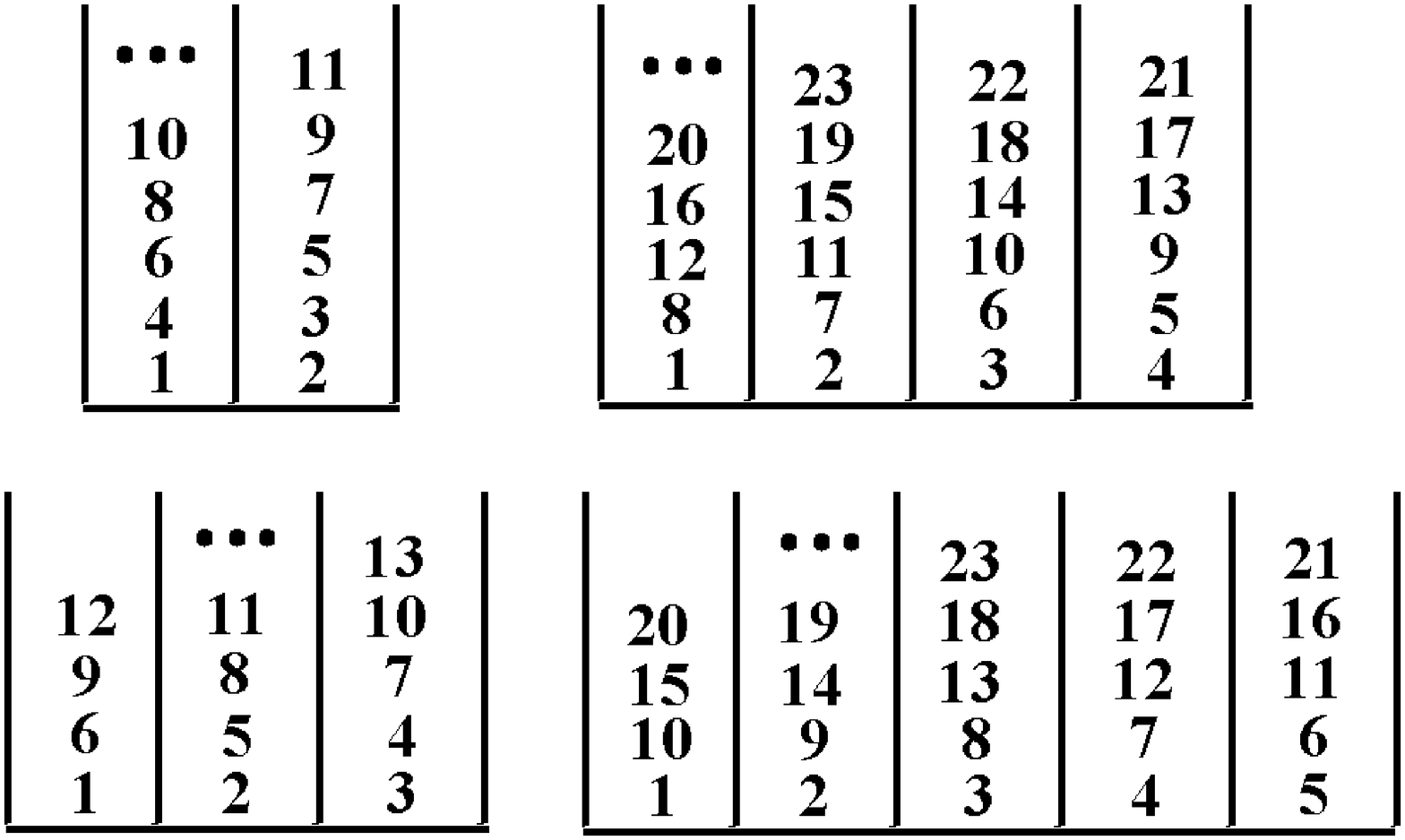, scale=0.33,angle=90}}}
\caption{The assignment of the first few jobs for the second solution for each number of machines. The positions of solutions are as in Figure \ref{ff1}.}
\label{ff2}
\end{figure}}

\subsection{Two machines}
We define two solutions, and show that the algorithm that uses them has competitive ratio at most $1.25$. The first solution has a set with all jobs of index $j$ such that $j  \equiv 1\bmod 4$, and the second set consists of all other jobs (every $j$ for which it holds that $j \bmod 4 \in \{0,2,3\}$). That is, one set is $\{1,5,9,13,\ldots\}$ and the other set is $\{2,3,4,6,7,8,\ldots\}$. The second solution has the following subsets. The first set consists of $1$ and all even indexes excluding $2$, while the second set has $2$ and all odd indexes excluding $1$. That is, the sets are $\{1,4,6,8,10,\ldots\}$ and $\{2,3,5,7,9,11,\ldots\}$.

\begin{theorem}
The competitive ratio of this algorithm for $m=2$ is at most $1.25$.
\end{theorem}
\begin{proof}
Recall that $C_1$ and $C_2$ are the total sizes of jobs for the first solution, and $L_1$ and $L_2$ are the total sizes of jobs for the second solution.
We consider $C_1$ and $L_2$, and afterwards we consider $C_2+L_1$. We will show the three inequalities: $C_1 \leq 1.25 \cdot \lambda$, $L_2 \leq 1.25 \cdot \lambda$, and $L_1+ C_2 \leq 2.5 \cdot \lambda$.

We have $C_1=\sum_{i=1}^{\infty} p_{4i-3}$.
Taking four occurrences of each job of this machine, we have the multi-set $\{1,1,1,1,5,5,5,5,9,9,9,9,\ldots\}$, whose total size is not smaller than that of the multi-set $\{1,1,1,1,2,3,4,5,6,7,8,9,\ldots\}$, by $p_{4i-3}\geq p_{4i-4} \geq p_{4i-5}\geq p_{4i-6}$ for $i\geq 2$.
Thus, $4\cdot C_1 -3\cdot p_1 \leq W$. (This can also be derived from applying Lemma \ref{theW} with $\beta=1$, $\alpha=4$, and $\gamma=1$).
By $p_1 \leq \lambda$ and $W \leq 2\lambda$, we have $$C_1 \leq \frac{W+3\cdot p_1}4 \leq \frac{2\lambda+3\lambda}4 =1.25 \cdot \lambda \ .$$

We have $L_2=p_2+p_3+\sum_{i=3}^{\infty} p_{2i-1}$.
The multi-set of jobs that corresponds to taking two occurrences of each job of $L_2$, of total size $2\cdot L_2$, is $\{2,2,3,3,5,5,7,7,9,9,\ldots\}$, and this total size does not exceed the total size of $\{2,2,3,3,4,5,6,7,8,9,\ldots\}$,  which is at most $W+p_2+p_3-p_1$, by using $p_{2i-1}\geq p_{2i-2}$ for $i\geq 3$. (Those arguments are simple, usually follow from Lemma \ref{theW} easily, and in what follows, we will not provide the full details in such arguments.) Thus, $2\cdot L_2 - p_3 - p_2 +p_1\leq W$.
By $p_3 \leq \frac{\lambda}2$, $p_1 \geq p_2$, and $W \leq 2\lambda$, we have $$L_2 \leq \frac{W+p_3+p_2-p_1}2 \leq \frac{W+p_3}2\leq \frac{2\lambda+\lambda/2}2 =1.25 \cdot \lambda \ .$$

Finally, consider  $4\cdot C_2+4 \cdot L_1=4\sum_{i=1}^{\infty} (p_{4i-2}+p_{4i-1}+p_{4i})+4\cdot p_1+4\sum_{i=2}^{\infty} p_{2i}$.
The multi-set for the two machines is
$$\{2,3,4,6,7,8,\ldots\}\cup\{1,4,6,8,10,\ldots\}=\{1,2,3,4,4,6,6,7,8,8,10,10,11,12,12,\ldots\} \ , $$ and by taking four occurrences,  we have $$\{1,1,1,1,2,2,2,2,3,3,3,3,4,4,4,4,4,4,4,4,6,6,6,6,6,6,6,6,7,7,7,7,8,8,8,8,8,8,8,8,$$ $$10,10,10,10,10,10,10,10,11,11,11,11,12,12,12,12,12,12,12,12,\ldots\} \ .$$
by replacing jobs by jobs that are not smaller, we have
$$\{1,1,1,1,1,2,2,2,2,2,3,3,3,3,3,4,4,4,4,4,5,5,5,5,5,6,6,6,6,6,7,7,7,7,7,8,8,8,8,8,$$ $$9,9,9,9,9,10,10,10,10,10,11,11,11,11,11,12,12,12,12,12,\ldots\} \ .$$
We get $4C_2+4L_1 \leq 5 W \leq 10 \cdot \lambda$.
\end{proof}
\subsection{Three machines}

We define two solutions, and show that the algorithm that uses them has competitive ratio at most $\frac 43$.

The first solution is defined as follows. Machine $1$ has the following jobs: $\{6k+1: k\geq 0\}$, that is, the set is $\{1,7,13,\ldots\}$. Machine $2$ has the following jobs: $\{6k+2: k\geq 0\}\cup \{6k+5: k\geq 0\}$, that is, the set is $\{2,5,8,11,14,17,20,\ldots\}$. Machine $3$ has the remaining jobs. Those are $\{6k+3: k\geq 0\}\cup \{6k+4: k\geq 0\}\cup \{6k+6: k\geq 0\}$, that is, the set is $\{3,4,6,9,10,12,15,\ldots\}$.

The second solution is defined as follows. Machine $i$ (for $1\leq i \leq 3$) has job $i$, and all jobs such that $j\geq 4$ and $(j+i) \bmod 3  = 1$. Thus, the set of machine $1$ is $\{1\}\cup\{3k: k\geq 2\}=\{1,6,9,12,\ldots\}$, the set of machine $2$ is $\{3k-1: k\geq 1\}=\{2,5,8,11,14,\ldots\}$, the set of machine $3$ is $\{3\}\cup\{3k-2: k\geq 2\}=\{3,4,7,10,13,\ldots\}$.

\begin{theorem}
The competitive ratio of this algorithm for $m=3$ is at most $\frac 43$.
\end{theorem}
\begin{proof}
We have that  $C_i$, for $i=1,2,3$, and $L_i$ for $i=1,2,3$ are the total sizes of jobs for the first and second solutions, respectively. That is,
$C_1=\sum_{i=1}^{\infty} p_{6i-5}$, $C_2=\sum_{i=1}^{\infty} (p_{6i-4}+p_{6i-1})$, $C_3=\sum_{i=1}^{\infty} (p_{6i-3}+p_{6i-2}+p_{6i})$,
$L_1=p_1+\sum_{i=2}^{\infty} p_{3i}$, $L_2=p_2+\sum_{i=2}^{\infty} p_{3i-1}$, and $L_3=p_3+\sum_{i=2}^{\infty} p_{3i-2}$.
We have $C_2=L_2$, and we will show the following three inequalities: $C_1, C_2 \leq \frac 43 \cdot \lambda$, $L_2, L_3 \leq \frac 43 \cdot \lambda$, and $6\cdot C_3 +6\cdot L_1 \leq 16 \cdot \lambda$.

We use Lemma \ref{theW} in the proof. Consider the first solution. By $W \leq 3\lambda$ and $p_1 \leq \lambda$, we have $$C_1 \leq \frac{W-p_1}6+p_1 =\frac{W}6+\frac 56 \cdot p_1 \leq \frac 43 \cdot \lambda \ .$$ We also have (by the same constraints and by $p_1 \geq p_2$) $$C_2 \leq \frac{W-p_1-p_2}3 + p_2 \leq \frac{W}3+ \frac{p_2}3 \leq \frac 43 \cdot \lambda \ .$$
Consider the second solution. We have $L_2=C_2$, and it is left to analyze $L_3$. We have (by using also the constraints $p_2\geq p_3 \geq p_4$) $$L_3 \leq \frac{W-P_4}3+p_3+p_4 \leq \frac{W}{3}+\frac 23 \cdot p_4\leq \frac 43 \cdot \lambda \ .$$

Finally, consider $6\cdot C_3+6 \cdot L_1$. The set of jobs for the two machines is
$$\{1,3,4,6,6,9,9,10,12,12,15,15,16,18,18,\ldots\} \ .$$ Taking this set such that every job occurs six times, and replacing jobs by jobs that are not smaller, we get
$$\{1,1,1,1,1,1,2,2,2,2,3,3,3,3,3,4,4,4,4,4,5,5,5,5,5,6,6,6,6,6,7,7,7,7,7,8,8,8,8,8,9,9,$$ $$9,9,9,10,10,10,10,10,11,11,11,11,11,12,12,12,12,12,13,13,13,13,13,14,$$ $$14,14,14,14,15,15,
15,15,15,16,16,16,16,16,17,17,17,17,17,18,18,18,18,18,\ldots\} \ .$$
For the jobs of indexes $6k+3,6k+3,6k+4,6k+6,6k+6$ (for $k\geq 1$, each index occurring six times such that for example $6k+3$ occurs $12$ times), it is replaced with $6k+1,6k+1,6k+1,6k+1,6k+1,
6k+2,6k+2,6k+2,6k+2,6k+2,
6k+3,6k+3,6k+3,6k+3,6k+3,
6k+4,6k+4,6k+4,6k+4,6k+4,
6k+5,6k+5,6k+5,6k+5,6k+5,
6k+6,6k+6,6k+6,6k+6,6k+6$.

For the jobs of indexes $1,3,4,6,6$ appearing six times, they were replaced with $1$, $1$, $1$, $1$, $1$, $1$, $2$, $2$, $2$, $2$, $3$, $3$, $3$, $3$, $3$, $4$, $4$, $4$, $4$, $4$, $5$, $5$, $5$, $5$, $5$, $6$, $6$, $6$, $6$, $6$.

We have $6\cdot C_3+6 \cdot L_1 \leq 5 \cdot W +  p_1 - p_2 \leq 16 \cdot \lambda$, by $p_1 \leq \lambda$, $p_2 \geq 0$, and $W \leq 3 \lambda$.
\end{proof}

\subsection{Four machines}

We define two solutions, and show that the algorithm that uses them has competitive ratio at most $1.375$.

The first solution is defined as follows. For $i=1,2$, machine $i$ has jobs $\{8k+i: k\geq 0\}$, that is, the sets are $\{1,9,17,25,\ldots\}$ for machine $1$ and $\{2,10,18,26,\ldots\}$ for machine $2$. Machine $3$ has the jobs $\{8k+3: k\geq 0\}\cup \{8k+6: k\geq 0\} \cup \{8k: k\geq 1\}=\{3,6,8,11,14,16,\ldots\}$. Machine $4$ has the jobs $\{8k+4: k\geq 0\}\cup \{8k+5: k\geq 0\} \cup \{8k+7: k\geq 0\}=\{4,5,7,12,13,15,\ldots\}$.

The second solution is defined as follows. Machine $i$ (for $1\leq i \leq 4$) has job $i$, and all jobs such that $j\geq 5$ and $(j+i) \bmod 4  = 1$. Thus, the set of machine $1$ is $\{1\}\cup\{4k: k\geq 2\}=\{1,8,12,16,20,\ldots\}$, the set of machine $2$ is $\{2\}\cup\{4k-1: k\geq 2\}=\{2,7,11,15,19,\ldots\}$, the set of machine $3$ is $\{3\}\cup\{4k-2: k\geq 2\}=\{3,6,10,14,18,\ldots\}$, and the set of machine $4$ is $\{4\}\cup\{4k-3: k\geq 2\}=\{4,5,9,13,17,\ldots\}$.

\begin{theorem}
The competitive ratio of this algorithm for $m=4$ is at most $1.375$.
\end{theorem}
\begin{proof}
We have that $C_i$ (for $i=1,2,3,4$) is the total sizes of jobs for the first solution, and $L_i$ is the total sizes of jobs for the second solution. We start with showing the following properties:
$C_1, C_2 \leq 1.375 \cdot \lambda$, $L_3, L_4 \leq 1.375 \cdot \lambda$.

We have: $C_1=\sum_{k=0}^{\infty} p_{8k+1}$, $C_2=\sum_{k=0}^{\infty} p_{8k+2}$, $L_3=p_3+\sum_{k=1}^{\infty} p_{4k+2}$, and $L_4=p_4+\sum_{k=1}^{\infty} p_{4k+1}$. We use Lemma \ref{theW}, $W \leq 4\cdot \lambda$, the sorted order of jobs, and $\lambda \geq p_1$.
Using the lemma with $\alpha=8$, $\beta=i$, and $\gamma=1$, we have for $i=1,2$ that $$C_i \leq p_i+\sum_{k=1}^{\infty} p_{8k+i} \leq p_i+\frac{W-P_i}8 \leq p_1+\frac{W-p_1}8 =\frac{W}8+\frac 78\cdot p_1 \leq \frac{\lambda}2+\frac 78 \cdot \lambda = 1.375 \cdot \lambda \ .$$

In order to analyze $L_3$ and $L_3$, we use the lemma with $\alpha=4$, $\beta=5-i$, and $\gamma=2$ for $i=3,4$, and find that $\sum_{k=2}^{\infty} p_{4k+5-i} \leq \frac{W-P_{9-i}}4$.

For $i=3$, we have $$L_3 \leq p_3+p_6+\frac{W-p_1-p_2-p_3-p_4-p_5-p_6}4 \ .$$ By $p_1 \geq p_2 \geq p_3$ and $p_4 \geq p_5 \geq p_6$, we have $L_3 \leq \frac{W+p_3+p_6}4$.
We will show that $p_3+p_6 \leq 1.5 \cdot \lambda$, and by $\lambda \geq \frac W4$, we find that $L_3 \leq 1.375 \cdot \lambda$. The required property holds since $p_3 \leq \lambda$ and $p_6 \leq \frac{\lambda}2$.

For $i=4$, we have $$L_4 \leq p_4+p_5+\frac{W-p_1-p_2-p_3-p_4-p_5}4 \ .$$ Since we have $\lambda \geq 2\cdot p_5$. We get $L_4 \leq \frac{W+3\cdot p_5}4 \leq 1.375 \cdot \lambda$.

Next, we consider the four sums $2\cdot C_3+L_1$, $2\cdot C_4+L_1$, $2 \cdot C_3+L_2$, and $C_4+2\cdot L_2$. We show that each one of them does not exceed $4\cdot \lambda$.
In the remaining cases, we use the sorted order of the jobs of indexes $8k+1,8k+2,\ldots,8k+8$.

Consider the first sum: $$2C_3+L_1=2\sum_{k=0}^{\infty} p_{8k+3}+2\sum_{k=0}^{\infty} p_{8k+6}+2\sum_{k=0}^{\infty} p_{8k+8}+\sum_{k=2}^{\infty} p_{4k}+p_1 \ ,$$

We have $p_1+2p_3+2p_6+3p_8 \leq P_8$. The remainder (after deducting the left hand side from the right hand side of the earlier expression) is $$2\sum_{k=1}^{\infty} p_{8k+3}+2\sum_{k=1}^{\infty} p_{8k+6}+2\sum_{k=1}^{\infty} p_{8k+8}+\sum_{k=3}^{\infty} p_{4k} \\ $$ $$= 2\sum_{k=1}^{\infty} p_{8k+3}+2\sum_{k=1}^{\infty} p_{8k+6}+2\sum_{k=1}^{\infty} p_{8k+8}+\sum_{k=1}^{\infty} p_{8k+8}+\sum_{k=1}^{\infty} p_{8k+4}
\leq W-P_8 \ .$$ Thus, $2C_3+L_1 \leq W \leq 4\lambda$.

\medskip
\medskip

Consider the second sum:
$$2C_4+L_1=2\sum_{k=0}^{\infty} p_{8k+4}+2\sum_{k=0}^{\infty} p_{8k+5}+2\sum_{k=0}^{\infty} p_{8k+7}+\sum_{k=2}^{\infty} p_{4k}+p_1 \ ,$$

 We have $p_1+2p_4+2p_5+2p_7 +p_8 \leq P_8$.  The remainder is $$2\sum_{k=1}^{\infty} p_{8k+4}+2\sum_{k=1}^{\infty} p_{8k+5}+2\sum_{k=1}^{\infty} p_{8k+7}+\sum_{k=3}^{\infty} p_{4k}\\$$ $$=  2\sum_{k=1}^{\infty} p_{8k+4}+2\sum_{k=1}^{\infty} p_{8k+5}+2\sum_{k=1}^{\infty} p_{8k+7}+\sum_{k=1}^{\infty} p_{8k+8}+\sum_{k=1}^{\infty} p_{8k+4} \leq W-P_8 \ . $$ Thus, $2C_4+L_1 \leq W \leq 4\lambda$.
\medskip
\medskip

Consider the third sum:
$$2C_3+L_2=2\sum_{k=0}^{\infty} p_{8k+3}+2\sum_{k=0}^{\infty} p_{8k+6}+2\sum_{k=0}^{\infty} p_{8k+8}+\sum_{k=2}^{\infty} p_{4k-1}+p_2 \ ,$$
We have $p_2+2p_3+2p_6+p_7 +2p_8 \leq P_8$.

The remainder is $$2\sum_{k=1}^{\infty} p_{8k+3}+2\sum_{k=1}^{\infty} p_{8k+6}+2\sum_{k=1}^{\infty} p_{8k+8}+\sum_{k=3}^{\infty} p_{4k-1} $$ $$= 2\sum_{k=1}^{\infty} p_{8k+3}+2\sum_{k=1}^{\infty} p_{8k+6}+2\sum_{k=1}^{\infty} p_{8k+8}+\sum_{k=1}^{\infty} p_{8k+3}+\sum_{k=1}^{\infty} p_{8k+7} \leq W-P_8 \ . $$ Thus, $2C_3+L_2 \leq W \leq 4\lambda$.
\medskip
\medskip

Consider the fourth sum (where the calculation is different since the structure of the sum is slightly different):
$$C_4+2\cdot L_2=\sum_{k=0}^{\infty} p_{8k+4}+\sum_{k=0}^{\infty} p_{8k+5}+\sum_{k=0}^{\infty} p_{8k+7}+2\sum_{k=2}^{\infty} p_{4k-1}+2\cdot p_2 \ . $$ We have $2p_2+p_4+p_5+3p_7 \leq P_7$.

The remainder is $$\sum_{k=1}^{\infty} p_{8k+4}+\sum_{k=1}^{\infty} p_{8k+5}+\sum_{k=1}^{\infty} p_{8k+7}+2\sum_{k=3}^{\infty} p_{4k-1} $$ $$= \sum_{k=1}^{\infty} p_{8k+4}+\sum_{k=1}^{\infty} p_{8k+5}+\sum_{k=1}^{\infty} p_{8k+7}+2\sum_{k=1}^{\infty} p_{8k+3}+2\sum_{k=1}^{\infty} p_{8k+7} \leq W-P_7 \ . $$
In this case, we use the set of indexes $\{8,9,10,11,12,13,14,15,\ldots\}$.
Thus, $2C_3+L_2 \leq W \leq 4\lambda$.

Recall that we have $C_1,C_2, L_3, L_4 \leq 1.375 \cdot \lambda$. Therefore, the competitive ratio holds.
\end{proof}
\subsection{Five machines}

We define the algorithm for $m=5$ as follows.
We define two solutions, and show that the algorithm that uses them has competitive ratio at most $\frac 43$.

The first solution is defined as follows. Machine $1$ has the following jobs: $\{10k+1: k\geq 0\}$, that is, the set is $\{1,11,21,\ldots\}$. Machine $2$ has the following jobs: $\{10k+2: k\geq 0\}$, that is, the set is $\{2,12,22,\ldots\}$.
Machine $3$ has the following jobs: $\{5k+3: k\geq 0\}$, that is, the set is $\{3,8,13,18,\ldots\}$.
Machine $4$ has the following jobs: $\{10k+4: k\geq 0\}\cup\{10k+7: k\geq 0\}\cup\{10k: k\geq 1\}$, that is, the set is $\{4,7,10,14,17,20,24\ldots\}$.
Machine $5$ has the following jobs: $\{10k+5: k\geq 0\}\cup\{10k+6: k\geq 0\}\cup\{10k+9: k\geq 0\}$, that is, the set is $\{5,6,9,15,16,19,\ldots\}$.

The second solution is defined as follows. Machine $i$ (for $1\leq i \leq 5$) has job $i$, and all jobs such that $j\geq 4$ and $(j+i) \bmod 5  = 1$. Thus, the set of machine $1$ is $\{1\}\cup\{5k: k\geq 2\}=\{1,10,15,\ldots\}$, the set of machine $2$ is $\{2\}\cup\{5k-1: k\geq 2\}=\{2,9,14,19,\ldots\}$, the set of machine $3$ is $\{5k-2: k\geq 2\}=\{3,8,13,18,\ldots\}$,
the set of machine $4$ is $\{4\}\cup\{5k-3: k\geq 2\}=\{4,7,12,17,18,\ldots\}$, and the set of machine $5$ is $\{5\}\cup\{5k-4: k\geq 2\}=\{5,6,11,16,21,\ldots\}$.

\begin{theorem}
The competitive ratio of this algorithm for $m=5$ is at most $1.4$.
\end{theorem}
\begin{proof}
We use the notation $C_i$ and  $L_i$ for $1 \leq i \leq 5$ again. In particular, we have $C_3=L_3$.
We will show the following inequalities: $C_1, C_2, C_3 \leq 1.4 \cdot \lambda$, $L_3, L_4, L_5 \leq 1.4 \cdot \lambda$. For the pairs of machines $i_1$ and $i_2$ where $i_1 \in \{4,5\}$ and $i_2 \in \{1,2\}$, we will consider the pairs $C_{i_1}$ and $L_{i_2}$, and by the pigeonhole principle we will have at least one of $\max\{L_1,L_2,L_3,L_4,L_5\} \leq 1.4 \cdot \lambda$ and
$\max\{C_1,C_2,C_3,C_4,C_5\} \leq  1.4 \cdot \lambda$. We use the same properties as in the earlier proofs.

Consider the first solution. We have $$C_1 \leq \frac{W-p_1}{10}+p_1 =\frac{W}{10}+0.9 \cdot p_1 \leq 1.4 \cdot \lambda \ .$$ We also have $C_2 \leq C_1$. Additionally, $$C_3 \leq \frac{W-P_3}{5}+p_3 \leq \frac{W}{5}+0.4 \cdot p_3 \leq 1.4 \cdot \lambda \ .$$

Consider the second solution. By $L_3=C_3$, it is left to analyze $L_4$ and $L_5$.

We have $$L_4 \leq \frac{W-P_7}5+p_4+p_7 \leq \frac{W}{5}+\frac{1}5 \cdot p_4+\frac 25 \cdot p_7\leq 1.4 \cdot \lambda \ ,$$ and $$L_5 \leq \frac{W-P_6}5+p_5+p_6 \leq \frac{W}{5}+\frac 45 \cdot p_7\leq 1.4 \cdot \lambda \ .$$

We will consider the following four sums $2C_4+2L_2$,  $2C_5+2L_2$, $10C_4+5L_1$, $10C_5+5L_1$. By showing $2C_i+2L_2 \leq 5.5 \cdot \lambda$ and $10C_i+5L_1 \leq 21 \cdot \lambda$,
we will find that if $\max\{L_1,L_2\} \geq 1.4 \cdot \lambda$, then  $\max\{C_1,C_2\} \leq 1.4 \cdot \lambda$ holds.

\medskip

In the first two proofs we will use $p_{10k+2} \geq p_{10k+4}$ for any $k\geq 0$.

We have $$2C_4+2L_2=2\sum_{k=0}^{\infty}p_{10k+4}+2\sum_{k=0}^{\infty}p_{10k+7}+2\sum_{k=0}^{\infty}p_{10k+10}+ 2\sum_{k=2}^{\infty}p_{5k-1} +2\cdot p_2 $$ $$=2\sum_{k=0}^{\infty}p_{10k+4}+2\sum_{k=0}^{\infty}p_{10k+7}+2\sum_{k=0}^{\infty}p_{10k+10} +2\sum_{k=0}^{\infty}p_{10k+9}+2\sum_{k=1}^{\infty}p_{10k+4}+2\cdot p_2  $$
$$ \leq 2\sum_{k=0}^{\infty}p_{10k+4}+2\sum_{k=0}^{\infty}p_{10k+7}+2\sum_{k=0}^{\infty}p_{10k+10} +2\sum_{k=0}^{\infty}p_{10k+9}+2\sum_{k=0}^{\infty}p_{10k+2}  \leq W\ . $$

Thus, $2C_4+2L_2 \leq 5\cdot \lambda$.

\medskip
\medskip

We have
$$2C_5+2L_2=2\sum_{k=0}^{\infty}p_{10k+5}+2\sum_{k=0}^{\infty}p_{10k+6}+2\sum_{k=0}^{\infty}p_{10k+9}+ 2\sum_{k=2}^{\infty}p_{5k-1} +2\cdot p_2 $$ $$ \leq 2\sum_{k=0}^{\infty}p_{10k+5}+2\sum_{k=0}^{\infty}p_{10k+6}+2\sum_{k=0}^{\infty}p_{10k+9} +2\sum_{k=0}^{\infty}p_{10k+9}+2\sum_{k=0}^{\infty}p_{10k+2}\ . $$

We use $$\sum_{k=0}^{\infty}p_{10k+9} = \sum_{k=0}^{\infty}p_{10k+19} +p_9 \leq \sum_{k=0}^{\infty}p_{10k+10} +p_9 \ , $$ (where only one of the four occurrences of the sum $\sum_{k=0}^{\infty}p_{10k+9}$ is replaced),
and get, $2C_5+2L_2 \leq W + p_9 \leq 5.5\cdot \lambda$.

\medskip
\medskip
We have
$$10C_4+5L_1=10\sum_{k=0}^{\infty}p_{10k+4}+10\sum_{k=0}^{\infty}p_{10k+7}+10\sum_{k=0}^{\infty}p_{10k+10}+ 5\sum_{k=2}^{\infty}p_{5k} +5\cdot p_1 $$ $$=10\sum_{k=0}^{\infty}p_{10k+4}+10\sum_{k=0}^{\infty}p_{10k+7}+10\sum_{k=0}^{\infty}p_{10k+10} +5\sum_{k=0}^{\infty}p_{10k+10}+5\sum_{k=1}^{\infty}p_{10k+5}+5\cdot p_1 \ . $$

We use  $$\sum_{k=1}^{\infty}p_{10k+5} +p_1 \leq \sum_{k=0}^{\infty}p_{10k+1}  \ , $$ and replace four of the five occurrences, and $$\sum_{k=1}^{\infty}p_{10k+5}=\sum_{k=0}^{\infty}p_{10k+15} \leq \sum_{k=0}^{\infty}p_{10k+10} \ , $$

The last expression is bounded from above by
$$10\sum_{k=0}^{\infty}p_{10k+4}+10\sum_{k=0}^{\infty}p_{10k+7}+15\sum_{k=0}^{\infty}p_{10k+10} +4\sum_{k=0}^{\infty}p_{10k+1}+\sum_{k=0}^{\infty}p_{10k+10}+ p_1 \ , $$
and we get, $10C_4+5L_1 \leq 4W + p_1 \leq 21\cdot \lambda$.

\medskip
\medskip

We have
$$10C_5+5L_1=10\sum_{k=0}^{\infty}p_{10k+5}+10\sum_{k=0}^{\infty}p_{10k+6}+10\sum_{k=0}^{\infty}p_{10k+9}+ 5\sum_{k=2}^{\infty}p_{5k} +5\cdot p_1 $$

An argument similar to the previous case shows that the last expression is bounded from above by
$$10\sum_{k=0}^{\infty}p_{10k+5}+10\sum_{k=0}^{\infty}p_{10k+6}+10\sum_{k=0}^{\infty}p_{10k+9} +5\sum_{k=0}^{\infty}p_{10k+10}+4\sum_{k=0}^{\infty}p_{10k+1}+\sum_{k=0}^{\infty}p_{10k+10}+ p_1 \ , $$
and we get, $10C_5+5L_1 \leq 4W + p_1 \leq 21\cdot \lambda$.
\end{proof}

\subsection{Discussion}
As we saw, the cases $m=2,3,4,5$ exhibit the advantage of using two-solution algorithms compared to one solution algorithms.
In this section we briefly discuss the limitation of our approach.
First, we note that the results for $m=4,5$ are not tight. We can see that in some cases the analysis can be improved.
In the proof for the case $m=4$ and the analysis for $L_3$, we can show that $p_3+p_6 \leq \frac 43 \cdot \lambda$, and thus $L_3 \leq \frac 43 \cdot \lambda$. Consider an optimal offline solution, and the assignment of (only) the first six jobs in it. If at least one of the first three jobs $1\leq i_1 \leq 3$ is assigned with another job out of the six $1 \leq i_2 \leq 6$ (where $i_1 \neq i_2$), we have $p_3+p_6 \leq p_{i_1}+p_{i_2} \leq \lambda \leq \frac 43 \cdot \lambda$. Otherwise, the three jobs $4,5,6$ are assigned together and $3\cdot p_6 \leq p_4+p_5+p_6 \leq \lambda$, and $p_3 \leq p_1 \leq \lambda$, so $p_3+p_6 \leq \lambda+ \frac{\lambda}3=\frac 43 \cdot \lambda$. A similar argument can be used in some other cases, for example, in the analysis of $L_4$ and $m=5$.

However, our upper bounds on the competitive for those algorithms (for $m=4,5$) cannot be improved. This can be seen by the following realization. There are $m+1$ jobs of size $\frac 12$, and the remaining jobs are of a very small fixed size $\eps>0$, and total size $\frac{m-1}2$. An optimal offline solution has completion time $1$, but solutions where the jobs of indexes $m$ and $m+1$ are assigned to machine $m$, and this machine receives approximately a fraction of $\frac 1m$ or a larger fraction of the remaining jobs, will have completion times of at least (approximately) $1.5-\frac{1}{2m}$. Since both solutions have this structure, this shows a lower bound on the makespan for $m=4,5$, and a lower bound of the competitive ratios which shows that the analysis is tight.

Since our four algorithms are constructed by similar methods which can be easily extended for any number of machines $m\geq 6$, one could still hope that our approach would give a two-solution algorithm of competitive ratio at most $1.5$ for any $m\geq 2$. However, this in not the case. As an example, consider the case $m=20$. Consider an input realization with one job of size $1$, thirty jobs of size $\frac{7}{30}$ each, and jobs of size $\eps$ of total size $12$. An optimal offline solution has one machine with a job of size $1$, ten machines with three jobs of size $\frac{7}{30}$, and a total size of $0.3$ with jobs of size $\eps$, and nine machines with jobs of size $\eps$ of total size $1$ per machine. A solution of the type of the first solutions (for $m=2,4$) would assign three jobs of size $\frac{7}{30}$ to the $m$th machine, and approximately $\frac{3}{40}$ of the total size of jobs of size $\eps$, for a load of $1.6$. A solution of the type of the second solutions would assign the jobs of size $1$ to the first machine, and approximately $\frac{1}{20}$ of the total size of jobs of size $\eps$, for a load of $1.6$. Generalizaing this for different values of $m$, one can get that the algorithm will not have competitive ratio below $\frac 53 - \frac{4}{3m}$ for even values of $m$, and below $\frac 53 - \frac{16}{9m-3}$ for odd values of $m$.
Thus, at present time it is unknown whether a two-solution algorithm can have an improved competitive ratio compared to one-solution algorithms for large values of $m$.

\section{Lower bounds for two solutions}
In this section, we prove lower bounds on the competitive ratio for two-solution ordinal algorithms. We show a lower bound of $1.25$ on the competitive ratio for $m=2$, which is tight, since we saw an algorithm matching this bound. For $m\geq 3$, we show a lower bound of $\frac 43$ on the competitive ratio, using two different constructions. For $m=3$, this matches the bound of the algorithm that we saw, and for $m=4,5$, the bound is close to the upper bounds. We prove the lower bound for $m=2$ first, then the one for $m\geq 4$, and finally the one for $m=3$, which requires a larger number of input realizations.

\begin{proposition}
Every two-solution ordinal algorithm for two machines has a competitive ratio of at least $1.25$.
\end{proposition}
\begin{proof}
Assume that there is an algorithm with two solutions and competitive ratio $r<1.25$. Consider the job sequence $<4,1,1,1,1>$. Since the optimal cost is $4$, the solution for which the cost is at most $r\cdot 4 <5$ must have the job subsets $\{1\}$ and $\{2,3,4,5\}$.

Consider the following job sequences. For each one of them we first show that the solution that we already considered is not the one with cost at most $r$ times the optimal cost. We will identify the subsets of the other solution for the first five jobs, and reach a contradiction.

\begin{enumerate}
\item For the job sequence $<3,2,2,1,0>$, the cost of an optimal offline solution is $4$, and the completion time of the already identified solution is $5$.
\item For the job sequence $<2,2,2,1,1>$, the cost of an optimal offline solution is $4$, and the completion time of the already identified solution is $6$.
\item For the job sequence $<2,1,1,1,1>$, the cost of an optimal offline solution is $3$, and the completion time of the already identified solution is $4$.
\end{enumerate}

The second solution has to obtain optimal solutions for all three job sequences, in order to avoid costs that are larger by $1$, since $(r-1)$ times the cost of an optimal solution is smaller than $1$. Given the first job sequence, the two subsets for the first four jobs are $\{1,4\}$ and $\{2,3\}$.
There are two possible structures for this solution, in terms of the subsets with five jobs.

For the solution with subsets $\{1,4\}$ and $\{2,3,5\}$, the total size for the second subset for the second job sequence is $5$, which gives a competitive ratio above $r$.
For the solution with subsets $\{1,4,5\}$ and $\{2,3\}$, the total size for the first subset for the third job sequence is $4$, which  gives a competitive ratio above $r$.
\end{proof}

\begin{proposition}
Every two-solution ordinal algorithm for $m\geq 4$ has a competitive ratio of at least $\frac 43$.
\end{proposition}
\begin{proof}
We use three input realizations as follows. All of them have $2m$ jobs of positive sizes.
\begin{enumerate}
\item An input with $2m$ jobs of positive (equal) sizes: $<1,1,\ldots,1>$.
\item An input with two jobs of size $3$ each, $m-4$ jobs of size $2$ each, and $m+2$ jobs of size $1$ each: $<3,3,2,\ldots,2,1\ldots,1>$. If $m=4$, there are no jobs of size $2$, but the construction is unchanged.
\item An input with one job of size $3$, $m-2$ jobs of size $2$, and $m+1$ of size $1$ each, that is, $<3,2,\ldots,2,1,\ldots,1>$.
\end{enumerate}
For the first input, an optimal offline solution has completion time $2$, by assigning two jobs to each machine. An optimal offline solution for the second input assigns the first job to the first machine, the second job to the second machine, there are $m-4$ machines with one job of size $2$ and a job of size $1$, and the two remaining machines have three jobs of size $1$ each.
An optimal offline solution for the third input assigns the first job to the first machine, there are $m-2$ machines with one job of size $2$ and a job of size $1$, and remaining machine has three jobs of size $1$ each.
The completion times for the last two solutions are equal to $3$.

Since all inputs have equal numbers of jobs with positive sizes, we consider the solutions in terms of the sets of $2m$ jobs.

In order to avoid a competitive ratio of $1.5$ for the first input, there has to be a solution with $m$ pairs of jobs, which are a partition of the $2m$ jobs.
This solution has a completion time of $4$ for the second input, since each one of the first two jobs is combined with one of the jobs of sizes at least $1$. This solution has a completion time of at least $4$ for the third input as well, since the first job is combined with another job. Thus, to avoid a competitive ratio of $\frac 43$, the second solution has each one of jobs $1$ and $2$ assigned to its own set, which holds due to the second input. This results in a schedule for the third input where jobs of integer sizes and total size $3m-5$ are assigned to $m-2$ machines, so there is at least one machine with completion time at least $4$, for a competitive ratio of at least $\frac 43$.
\end{proof}

\begin{proposition}
Every two-solution ordinal algorithm for three machines has a competitive ratio of at least $\frac 43$.
\end{proposition}
\begin{proof}
Assume that there is an algorithm with two solutions with competitive ratio $r<\frac 43\approx 1.33333$.
Consider the next four input realizations with seven specified jobs provided as sequences (such that the remaining jobs have zero sizes).

\begin{enumerate}
\item For the job sequence $<3,3,1,1,1,0,0>$, the optimal offline completion time is $3$.
\item For the job sequence $<2,1,1,1,1,0,0>$, the optimal offline completion time is $2$.
\item For the job sequence $<2,2,2,1,1,1,0>$, the optimal offline completion time is $3$.
\item For the job sequence $<3,1,1,1,1,1,1>$, the optimal offline completion time is $3$.
\end{enumerate}

For all these sequences, at least one of the two solutions has to be optimal, since otherwise the makespan is larger by $1$ than an optimal solution, and the competitive ratio is at least $\frac 43$.

Thus, given the first sequence, there is one solution such that the sets for the first five jobs are $\{1\}$, $\{2\}$, and $\{3,4,5\}$. For the second and third sequence, the loads of the third machine and this solution are at least $3$ and $4$ respectively, and therefore the second solution has to be optimal for them (in order to avoid a competitive ratio of $\frac 43$). For the second sequence, an optimal solution for the first five jobs has the set $\{1\}$ and two sets with two jobs each out of $\{2,3,4,5\}$. For the third sequence, an optimal solution for the first six jobs has three sets, each having one job out of $\{1,2,3\}$ and one job out of $\{4,5,6\}$. Thus, for the first six jobs, one set is $\{1,6\}$, and the other two sets are either $\{2,4\}$ and $\{3,5\}$, or $\{2,5\}$ and $\{3,4\}$. For the fourth sequence, both possible second solutions have load of $4$ for the first machine. Thus, the first solution has to be optimal for it. For that, the set of the second machine and the first seven jobs has to be $\{2,6,7\}$ and the other two sets are still $\{1\}$ and $\{3,4,5\}$.

We consider two additional job sequences, which are input realizations with eight jobs of positive sizes.

\begin{enumerate}
\item[5.] For the job sequence $<6,6,1,1,1,1,1,1>$, the optimal offline completion time is $6$.
\item[6.] For the job sequence $<3,3,3,3,3,1,1,1>$, the optimal offline completion time is $6$.
\end{enumerate}

For these sequences (the fifth and sixth sequences), in order to avoid a competitive ratio of at least $\frac 43$, there has to be a solution that has completion time of at most $7$. For the first solution and the fifth sequence, the load of the second machine is $8$. For the first solution and the sixth sequence, the load of the third machine is $9$. Thus, the second solution has to have loads of at most $7$ for all machines, and both jobs sequences. We consider all options for this solution. The sets for the first six jobs are $\{1,6\}$, $\{2,j_1\}$ and $\{3,j_2\}$, where $j_1\neq j_2$, and $\{j_1,j_2\}=\{4,5\}$.
Given the fifth sequence, in order to avoid loads of $8$, jobs $7$ and $8$ belong to the same set as job $3$, and the sets for the eight first jobs are    $\{1,6\}$, $\{2,j_1\}$ and $\{3,j_2,7,8\}$. However, we find the for the sixth sequence, the load of the third machine will be $8$, and the competitive ratio is $\frac 43$, a contradiction.
\end{proof}

\section{Lower bounds for a single solution}
In this section, we show improved lower bounds on the competitive ratio for one-solution algorithms, for any number of machines $m$ such that $m\in \{5,6,\ldots,17\}$. The method is not very different from that of \cite{LSV96}, and the improved values result from a more careful analysis for a small number of machines, and from using an additional simple input. The main goal here is to show an improved lower bound on the competitive ratio for one-solution algorithms for $m=5$, thus showing that our two-solution algorithm has a better performance than any one-solution algorithm. A result of \cite{LSV96} states that in order to obtain a competitive ratio not exceeding $1.5$, the first $2m$ jobs are assigned such that machine $i$ has the jobs of indexes $i$ and $2m+1-i$. We will assume that the algorithm satisfies this property. We consider inputs with at least $2m$ jobs, and in the case where our construction yields a lower bound above $1.5$ on the competitive ratio, we truncate it to $1.5$ (the only case where we do this is $m=17$). A lower bound of $1.5$ on the competitive ratio was proved in \cite{LSV96} for $m\geq 18$, see Table \ref{t1}, where the known results for $m=2,3,4$ are given as well (the bounds for $m=2,3$ are tight). The table also contains our improved lower bounds.

Let $N$ be a large positive integer. Let $n \geq N$ be such that $n$ is divisible by $20!$. All inputs will consist of $n$ jobs of positive sizes (so we will discuss only these jobs). There are three types of inputs. The first two types are inspired by those of the construction of \cite{LSV96}, and the third one is simple. All inputs will be such that an optimal solution can have makespan no larger than $1$, as we will show.

The first class of inputs (or realizations) is defined as follows, for any $i$ such that $1 \leq i \leq m-1$. The input consists of $i$ jobs of size $1$ each, and $n-i$ jobs of size $\frac{m-i}{n}$ each (where the size $\frac{m-i}{n}$ is called the smaller size for this input).

The second class of inputs is defined as follows, for any $i$ such that $2 \leq i \leq m$. The input consists of $2m-i+1$ jobs of size $\frac 12$ each, and $n-2m+i-1$ jobs of size $\frac{i-1}{2n}$ each (where the size $\frac{i-1}{2n}$ is called the smaller size).

The third input consists of $n$ jobs, each of size $\frac mn$.

\begin{lemma}
Each one of the inputs has an optimal offline solution of cost not exceeding $1$.
\end{lemma}
\begin{proof}
For the first class of inputs, there will be $i$ machines with one job of size $1$ each. Since $m-i \leq m-1 < 20$, we have that $n$ is divisible by $m-i$. Every machine out of the remaining $m-i$ machines will have at most $\frac{n}{m-i}$ jobs of size $\frac{m-i}{n}$ each, for loads not larger than $1$. The number of such jobs is $n-i$, which is smaller than $n$, so all jobs are assigned. Thus, it is possible to assign all $n$ jobs without exceeding loads of $1$ for the machines.

For the second class of inputs, there will be $m-i+1 \leq m-1$ machines with two jobs of size $\frac 12$ each. The remaining $i-1 \geq 1$ machines get one job of size $\frac 12$ each, and additionally, at most $\frac{n}{i-1}$ jobs of size $\frac{i-1}{2n}$ each. In this case $i-1 \leq m-1$, and therefore $n$ is divisible by $i-1$. Every machine has completion time of at most $1$, and there is space for $n$ jobs, while their number is $n-2m+i-1 \leq n-m-1 <n$.

Since $1 \leq m-i \leq m-1 < 20$, we have that $n$ is divisible by $m-i$. Every machine out of the remaining $m-i$ machines will have at most $\frac{n}{m-i}$ jobs of size $\frac{m-i}{n}$ each. The number of such jobs is smaller than $n$, while it is possible to assign $n$ such jobs without exceeding completion time of $1$ for the machines.

For the last input, since $m \leq 17$, $n$ is divisible by $m$. Thus, it is possible to assign $\frac nm$ jobs to every machine. Since the size of every job is $\frac mn$, the completion time of each machines is $1$, and thus the makespan is $1$.
\end{proof}

Let $n_k$ be the number of jobs (out of the first $n$ jobs) assigned by the algorithm to machine $k$ (for $k=1,2,\ldots,m$), such that $\sum_{k=1}^m n_k=n$. In each of our inputs, the last $n-2m$ jobs have equal sizes, while the first $2m$ jobs are assigned as described above (two to each machine). Thus, for the first two classes of inputs, out of $n_k$ jobs assigned to machine $k$, at least $n_k-2$ jobs are of the corresponding smaller size (depending on the input). For the last input, all jobs have the same size, so there are $n_k>n_k-2$ jobs of one size.
Let $\nu_k=n_k-2$. We have $\sum_{k=1}^m \nu_i = n-2m$. Let $R$ be the competitive ratio of the algorithm. We will use the property that loads cannot exceed $R$ (since the optimal offline completion time is at most $1$ for all cases).

Consider the first input with a fixed $i$, and consider machine $i$ (since the completion time for the schedule is at least its completion time).
This machine has one job of size $1$, and at least $n_i-2=\nu_i$ jobs of the smaller size.
The completion time of machine $i$ is at least $1+\nu_i \cdot \frac{m-i}{n}$. Since an optimal solution has completion time not exceeding $1$, we get $1+\nu_i \cdot \frac{m-i}{n} \leq R$, which holds for $i\leq m-1$.

Consider the second input with a fixed $i$, and consider machine $i$.
This machine has two job of size $\frac 12$ each, and at least $n_i-2=\nu_i$ jobs of the smaller size.
The completion time of machine $i$ is at least $1+\nu_i \cdot \frac{i-1}{2n}$. Since an optimal solution has completion time not exceeding $1$, we get $1+\nu_i \cdot \frac{i-1}{2n} \leq R$, which holds for $i\geq 2$.

For the third input, we have for every machine $i$, $\nu_i \cdot \frac mn \leq R$.
Letting $\rho_i=\frac{\nu_i}n$, we have the three types of constraints $1+\rho_i \cdot (m-i) \leq R$,
$1+ \rho_i \cdot (i-1)/2\leq R$, and $\rho_i \cdot m\leq R$.
Rearranging again, we have $\rho_i  \leq \frac{R-1}{m-i}$,
$\rho_i \leq \frac{2(R-1)}{i-1}$, and $\rho_i \leq \frac Rm$.

This gives us a linear program. In order to show the lower bound, we will use only a part of the constraints, such that there is only one used constraint per machine, which is defined as a function of the index. For $1\leq i \leq \lfloor \frac{2m}3 \rfloor$, we use the constraint of the first class of inputs. For $\lceil \frac{2m}3 \rceil +1 \leq i \leq m$, we use the constraint of the second class of inputs. If $m$ is divisible by $3$, there are no other constraints. If $m=3t+1$ for a positive integer $t$, the constraint of the third input is used for machine $2t+1$. If $m=3t+2$ for a positive integer $t$, the constraint of the third input is used for machine $2t+2$. By taking the sum of constraints, the left hand side is equal to $1-\frac{2m}n$ which is close to $1$ for large values of $N$ and $n\geq N$ (since $m\leq 17$).

Since we are mostly interested in the case $m=5$, we write the constraints for this case in detail. In this case, the first constraint is used for machines $1,2,3$, the second constraint for machine $5$, and the third constraint for machine $4$.
\begin{equation*}
\rho_1 \leq \frac{R-1}4  \ , \ \ \
\rho_2 \leq \frac{R-1}3 \ , \ \ \
\rho_3 \leq \frac{R-1}2 \ , \ \ \
\rho_4 \leq \frac R5 \ , \ \  \
\rho_5 \leq \frac{R-1}2 \ .
\end{equation*}

Taking the sum, and letting $N$ grow to infinity, we have $R \cdot (\frac 14+\frac 13+\frac 12+\frac 15+\frac 12) \geq  1+ (\frac 14+\frac 13+\frac 12+\frac 12)$, which gives $R \geq \frac{155}{107}\approx 1.448598$.

We conclude this section with the following theorem.
\begin{theorem}
Every single-solution algorithm has a competitive ratio of at least $\frac{155}{107}\approx 1.448598$ for $m=5$, and as stated in the third column of Table \ref{t1} for $6 \leq m \leq 17$.
\end{theorem}

{\begin{table}[h!] {\renewcommand{\arraystretch}{1.2   }
\small
$$
\begin{array}{||c||c||c||}
\hline \hline
 m & \mbox{previous lower bound \cite{LSV96}}  &   \mbox{new lower bound}  \\
\hline
2 & \frac 43 \approx 1.3333 & - \\ \hline
3 & 1.4 & - \\ \hline
4 & 1.4375 & - \\ \hline
5 & 1.3704 & \frac{155}{107}\approx 1.448598 \\ \hline
6 & 1.4539 & \frac{191}{131}\approx 1.458015 \\ \hline
7 & 1.4542 & \frac{1127}{767}\approx  1.469361\\ \hline
8 & 1.4485 & \frac{2278}{1543}\approx 1.476344 \\ \hline
9 & 1.4775 & \frac{2593}{1753}\approx  1.479178 \\ \hline
10 & 1.4776 & \frac{257}{173}\approx 1.485549 \\ \hline
11 & 1.4744 & \frac{341}{229}\approx 1.489082 \\ \hline
12 & 1.4891 & \frac{1873}{1257}\approx  1.490055 \\ \hline
13 & 1.4892 & \frac{7449}{4985}\approx  1.494282 \\ \hline
14 & 1.4872 & \frac{201739}{134815}\approx 1.496413 \\ \hline
15 & 1.4961 & \frac{217183}{145111}\approx 1.496668 \\ \hline
16 & 1.4961 & \frac{2027686}{1352011}\approx  1.499755\\ \hline
17 & 1.4947 &  1.5 \\ \hline
m \geq 18 & 1.5 & - \\ \hline

\hline
\end{array}
$$}
\caption{\label{t1} Lower bounds for single-solution algorithms, as a function of the number of machines $m$.}
\end{table}
}

\bibliographystyle{abbrv}

\end{document}